\documentclass[twocolumn,10pt]{article}
\usepackage[left=1.5cm,right=1.5cm,top=3.75cm,bottom=3.75cm]{geometry}

\usepackage{graphicx}

\usepackage{authblk}
\usepackage{amsmath}
\usepackage{physics}
\usepackage{amsfonts}

\usepackage{array}

\newcolumntype{+}{!{\vrule width 2pt}}

\newlength\savedwidth

\newcommand\thickhline{\noalign{\global\savedwidth\arrayrulewidth\global\arrayrulewidth 2pt}
\hline
\noalign{\global\arrayrulewidth\savedwidth}}

\usepackage{hyperref}

\usepackage[
sorting=none, 
defernumbers=true,
style=numeric-comp,
sortcites=true 
]{biblatex}

\AtEveryBibitem{\clearlist{language}}

\newcommand{\newcaption}[2]{
    \caption[#1]{\textbf{#1.} #2}}

\usepackage{soul, xcolor}

\usepackage{amsmath}
\usepackage{amsthm}

\newtheorem*{theorem*}{Theorem}

\title{Trophic structure predicts seizure propagation in brain network models}

\author[1]{Peter Kissack}
\author[2]{Catherine Drysdale}
\author[1]{Samuel Johnson}
\affil[1]{School of Mathematics, College of Engineering and Physical Sciences, University of Birmingham, Birmingham, UK}
\affil[2]{School of Mathematical Sciences, Lancaster University, Lancaster, UK}
\date{}

\begin{document}

\twocolumn[
\begin{@twocolumnfalse}
    \maketitle

\begin{abstract}
Epilepsy is widely regarded as a disorder driven by connectivity in the brain. We use a model of seizure dynamics on directed networks to investigate how structural properties affect seizure propensity. We find that properties such as trophic coherence, spectral radius, strong connectivity and non-normality are closely related to seizure propensity, and present a proof of a theoretical relationship between spectral radius and cycle structure. Our simulated results are robust to the kind of coupling used in the model and become stronger as network size is increased. They suggest that the overall directionality of information processing in the brain may be related to a propensity for epileptic seizures.
\end{abstract}
\vspace*{1em}

\end{@twocolumnfalse}
]

\section{Introduction} \label{sec:Intro}

Epilepsy is a neurological condition which is understood to be a disorder of pathological networks in the brain~\cite{berg_revised_2010,fisher_operational_2017}.
It follows that the topological and graph theoretic properties of such networks may provide key insights into mechanisms of seizure onset and propagation.
To this end, several conventional graph theoretical measures of topological features, such as notions of path length, clustering and node centrality, have been applied to networks inferred from various brain imaging data~\cite{van_diessen_functional_2013, diessen_brain_2014, royer_epilepsy_2022,pedersen_brain_2024,cottin_eeg_2026}.

Amongst the aspects of directed network topology with potential to impact a brain network's propensity to propagate seizure activity are hierarchical flow, cycles and strong connectivity, previously examined by the means of first transitive component (FTC)~\cite{benjamin_phenomenological_2012} and the existence of strongly-connected components and cycles~\cite{schmidt_dynamics_2014}.
The notions of trophic levels and trophic incoherence~\cite{johnson_trophic_2014,johnson_looplessness_2017,mackay_how_2020}, as well as related concepts such as the non-normality and spectrum of the adjacency matrix~\cite{rodgers_strong_2023,drysdale_connection_2025}, encapsulate these network properties at a global level. However, few studies have applied these to epileptic brain networks, despite their relevance to spreading dynamic behaviour on networks. 

In this paper we discuss these features of networks and how they relate to dynamical systems, using a phenomenological model of an epileptic brain network.
We simulate a phenomenological model of seizure propagation on directed networks and compare trophic incoherence, adjacency spectral radius and some related quantities describing asymmetry and strong-connectivity to a simulation outcome measure quantifying seizure propensity in a network model.

We compare results simulated using additive and diffusive coupling, a modelling choice which informs the influence of network interactions on model dynamics, and may significantly impact simulation outcomes~\cite{lopes_role_2023}. We also simulate both 20 node networks, of a similar size to those derived from standard, clinical EEG recordings, and larger, 128 node networks representing a higher resolution model of the brain, capable of exhibiting a higher complexity of network interactions.

We present a theoretical result relating the adjacency spectral radius to cycles and strong-connectivity and demonstrate that trophic structure and strong-connectivity are highly correlated to simulation outcomes in a network model of epileptic seizures, across different coupling types and network sizes.

\section{Results}

\begin{figure*}[!h]
 \centering\includegraphics[width=6in]{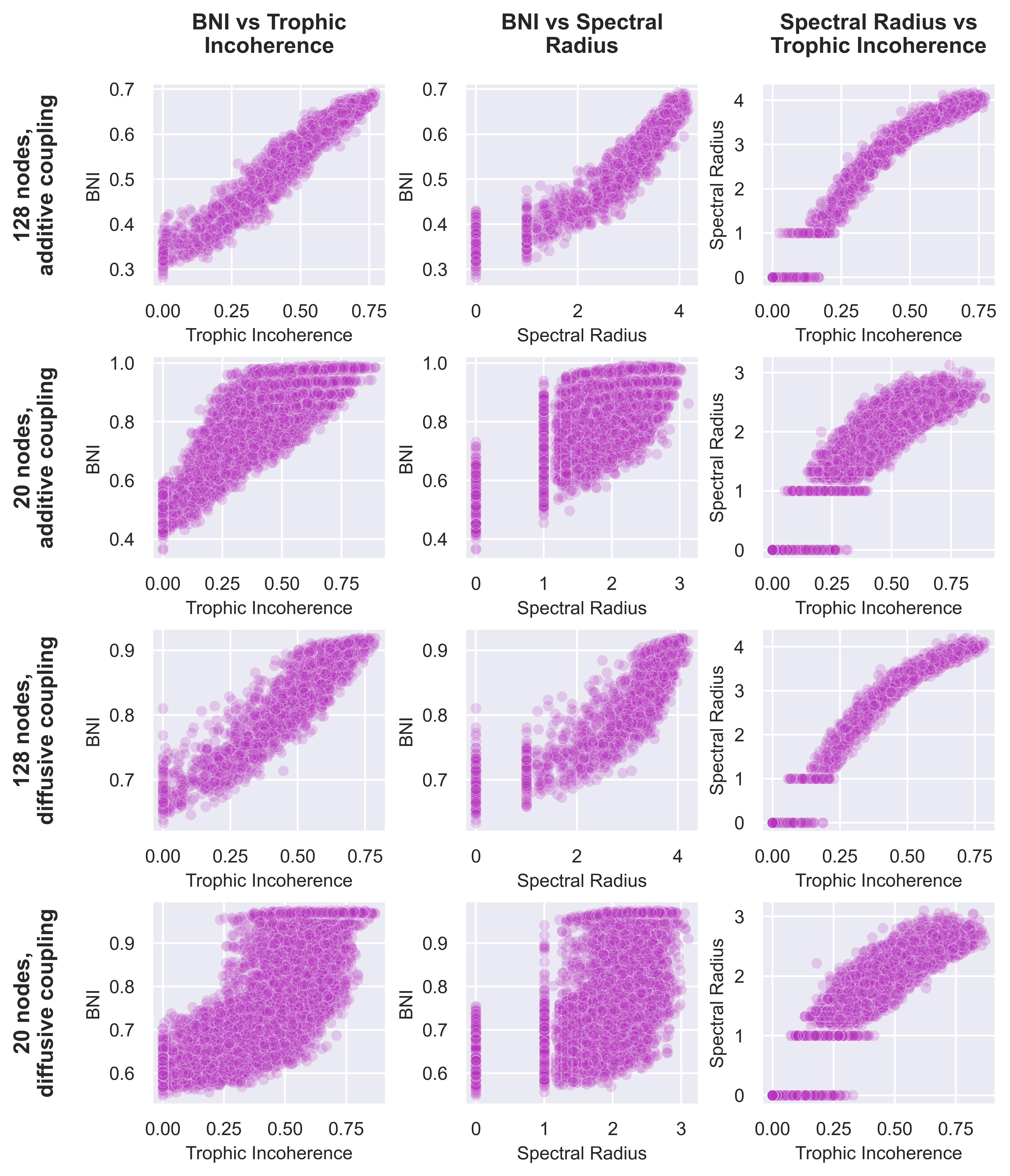}
\newcaption{Trophic Incoherence, Spectral Radius and BNI}{
Pairwise scatter plots of trophic incoherence, BNI and spectral radius, for 128- and 20-node networks sampled using the generalised preferential preying algorithm and simulated with additive and diffusive coupling.}
\label{fig:TrophicSpectScatter}
\end{figure*}

\begin{table*}[!ht]
\centering
\newcaption{Larger, additively-coupled network simulation correlations}{Spearman rank correlations between graph features and simulation outcomes from 2000 128-node networks sampled using the generalised preferential preying algorithm, simulated with additive coupling. In the first column Brain Network Ictogenicity (BNI), a metric reflecting the propensity of simulated seizure-like behaviour to spread in a network, is compared against a selection of global network properties, and properties of the node in which seizure-like behaviour most frequently emerges across simulations of the network.}
\label{tab:corrs_128_additive}
\begin{tabular}{|l+l|l|l|l|l|l|l|}
\hline
 & A & B & C & D & E & F & G  \\
\thickhline
A. BNI &&&&&&& \\ \hline
B. spectral radius & 0.957 &&&&&&
\\ \hline
C. trophic incoherence & 0.969 & 0.975 &&&&&
\\ \hline
D. largest strongly-connected component & 0.927 & 0.902 & 0.951 &&&&
\\ \hline
E. non normality & -0.766 & -0.731 & -0.756 & -0.711 &&&
\\ \hline
F. first transitive component size & -0.223 & -0.235 & -0.227 & -0.209 & 0.176 &&
\\ \hline
G. trophic level of the modal start node & -0.205 & -0.189 & -0.205 & -0.214 & 0.100 & -0.052 &
\\ \hline
H. reach of the modal start node & 0.880 & 0.815 & 0.845 & 0.870 & -0.621 & -0.157 & -0.430
\\ \hline
\end{tabular}
\end{table*}

\begin{table*}[!ht]
\centering
\newcaption{Smaller, additively-coupled network simulation correlations}{Spearman rank correlations between graph features and simulation outcomes from 10000 20-node networks sampled using the generalised preferential preying algorithm, simulated with additive coupling. In the first column Brain Network Ictogenicity (BNI), a metric reflecting the propensity of simulated seizure-like behaviour to spread in a network, is compared against a selection of global network properties, and properties of the node in which seizure-like behaviour most frequently emerges across simulations of the network.}
\label{tab:corrs_20_additive}
\begin{tabular}{|l+l|l|l|l|l|l|l|}
\hline
 & A & B & C & D & E & F & G  \\
\thickhline
A. BNI &&&&&&& \\ \hline
B. spectral radius & 0.771 &&&&&& 
\\ \hline
C. trophic incoherence & 0.855 & 0.921 &&&&& 
\\ \hline
D. largest strongly-connected component & 0.835 & 0.871 & 0.942 &&&&
\\ \hline
E. non normality & -0.562 & -0.597 & -0.649 & -0.591 &&&
\\ \hline
F. first transitive component size & 0.002 & -0.138 & -0.120 & -0.097 & 0.058 &&
\\ \hline
G. trophic level of the modal start node & -0.722 & -0.705 & -0.747 & -0.707 & 0.354 & -0.016 &
\\ \hline 
H. reach of the modal start node & 0.611 & 0.683 & 0.699 & 0.709 & -0.481 & -0.003 & -0.750
\\ \hline
\end{tabular}
\end{table*}

\begin{table*}[!ht]
\centering
\newcaption{Larger, diffusively-coupled network simulation correlations}{Spearman rank correlations between graph features and simulation outcomes from 2000 128-node networks sampled using the generalised preferential preying algorithm, simulated with diffusive coupling. In the first column Brain Network Ictogenicity (BNI), a metric reflecting the propensity of simulated seizure-like behaviour to spread in a network, is compared against a selection of global network properties, and properties of the node in which seizure-like behaviour most frequently emerges across simulations of the network.}
\label{tab:corrs_128_diffusive}
\begin{tabular}{|l+l|l|l|l|l|l|l|}
\hline
 & A & B & C & D & E & F & G  \\
\thickhline
A. BNI &&&&&&& \\ \hline
B. spectral radius & 0.905 &&&&&&
\\ \hline
C. trophic incoherence & 0.922 & 0.979 &&&&&
\\ \hline
D. largest strongly-connected component & 0.915 & 0.909 & 0.951 &&&&
\\ \hline
E. non normality & -0.726 & -0.747 & -0.765 & -0.722 &&&
\\ \hline
F. first transitive component size & -0.226 & -0.219 & -0.223 & -0.214 & 0.165 &&
\\ \hline
G. trophic level of the modal start node & -0.016 & -0.020 & -0.021 & -0.009 & 0.046 & 0.156 &
\\ \hline
H. reach of the modal start node & 0.882 & 0.784 & 0.809 & 0.826 & -0.617 & -0.298 & -0.232 
\\ \hline

\end{tabular}
\end{table*}

\begin{table*}[!ht]
\centering
\newcaption{Smaller, diffusively-coupled network simulation correlations}{Spearman rank correlations between graph features and simulation outcomes from 10000 20-node networks sampled using the generalised preferential preying algorithm, simulated with diffusive coupling. In the first column Brain Network Ictogenicity (BNI), a metric reflecting the propensity of simulated seizure-like behaviour to spread in a network, is compared against a selection of global network properties, and properties of the node in which seizure-like behaviour most frequently emerges across simulations of the network.}
\label{tab:corrs_20_diffusive}
\begin{tabular}{|l+l|l|l|l|l|l|l|}
\hline
 & A & B & C & D & E & F & G  \\
\thickhline
A. BNI &&&&&&& \\ \hline
B. spectral radius & 0.727 &&&&&&
\\ \hline
C. trophic incoherence & 0.808 & 0.923 &&&&& 
\\ \hline
D. largest strongly-connected component & 0.784 & 0.869 & 0.939 &&&& 
\\ \hline
E. non normality & -0.502 & -0.595 & -0.645 & -0.587 &&&
\\ \hline
F. first transitive component size & -0.001 & -0.113 & -0.090 & -0.068 & 0.053 &&
\\ \hline
G. trophic level of the modal start node & -0.020 & -0.023 & -0.035 & -0.023 & 0.041 & 0.456 &
\\ \hline
H. reach of the modal start node & 0.728 & 0.600 & 0.672 & 0.675 & -0.441 & -0.356 & -0.322
\\ \hline
\end{tabular}
\end{table*}

Fig~\ref{fig:TrophicSpectScatter} depicts the pairwise relationships between brain network ictogenicity (BNI), a model outcome describing the prevalence of the seizure state in a series of simulations of the network, spectral radius of the directed adjacency matrix, and trophic incoherence, in four different model scenarios, modelling additively or diffusively coupled networks with 20 or 128 nodes.

We observe five further features derived from network structure and simulation outcomes, namely the size of the largest strongly-connected component, network non-normality, size of the FTC, and the trophic level and reach of the node on which the first seizure activity most frequently emerges over multiple simulations, termed the modal seizure start node. Further details on all features are included in Methods. In addition to reporting on these results we note the discontinuous distribution of the adjacency spectral radius, and present a theorem which accounts for the discontinuity and provides an intuition for the network properties it encapsulates.

\subsection{Simulation results} \label{sec:simresults}

Several of the included network and simulation-derived features correlate to BNI, to different extents in the different model scenarios. Table~\ref{tab:corrs_128_additive} presents Spearman's rank correlation coefficients between graph features and simulation outcomes, including BNI, for the simulations of 2000 128-node networks with additive coupling, sampled using the generalised preferential preying algorithm, and Table~\ref{tab:corrs_20_additive} presents the same for 10000 20-node networks.
Tables~\ref{tab:corrs_128_diffusive} and \ref{tab:corrs_20_diffusive} show respectively the same for simulations with diffusive coupling.

Coupling type does not affect network sampling: features solely encapsulating network structure, rather than simulation outcomes, have the same distributions in the additive and diffusive samples for each network size and sampling method.

Fig~\ref{fig:TrophicSpectScatter} shows scatter plots of trophic incoherence, spectral radius and BNI for both 20 and 128-node networks.
For the larger networks, there are very strong positive relationships between spectral radius, trophic incoherence and BNI, which are present to a lesser extent in the smaller networks.
Results for 20-node graphs sampled using the Erd\H{o}s-Renyi algorithm are shown in Supplementary Tables~S1-S2 and Supplementary Fig~S1; these are consistent with the presented results over a narrower distribution of trophic incoherence, as expected.
Supplementary Figs~S2-S7 show scatterplot matrices of all features for each simulation configuration.

\begin{figure*}[!h]
\centering\includegraphics[width=6in]{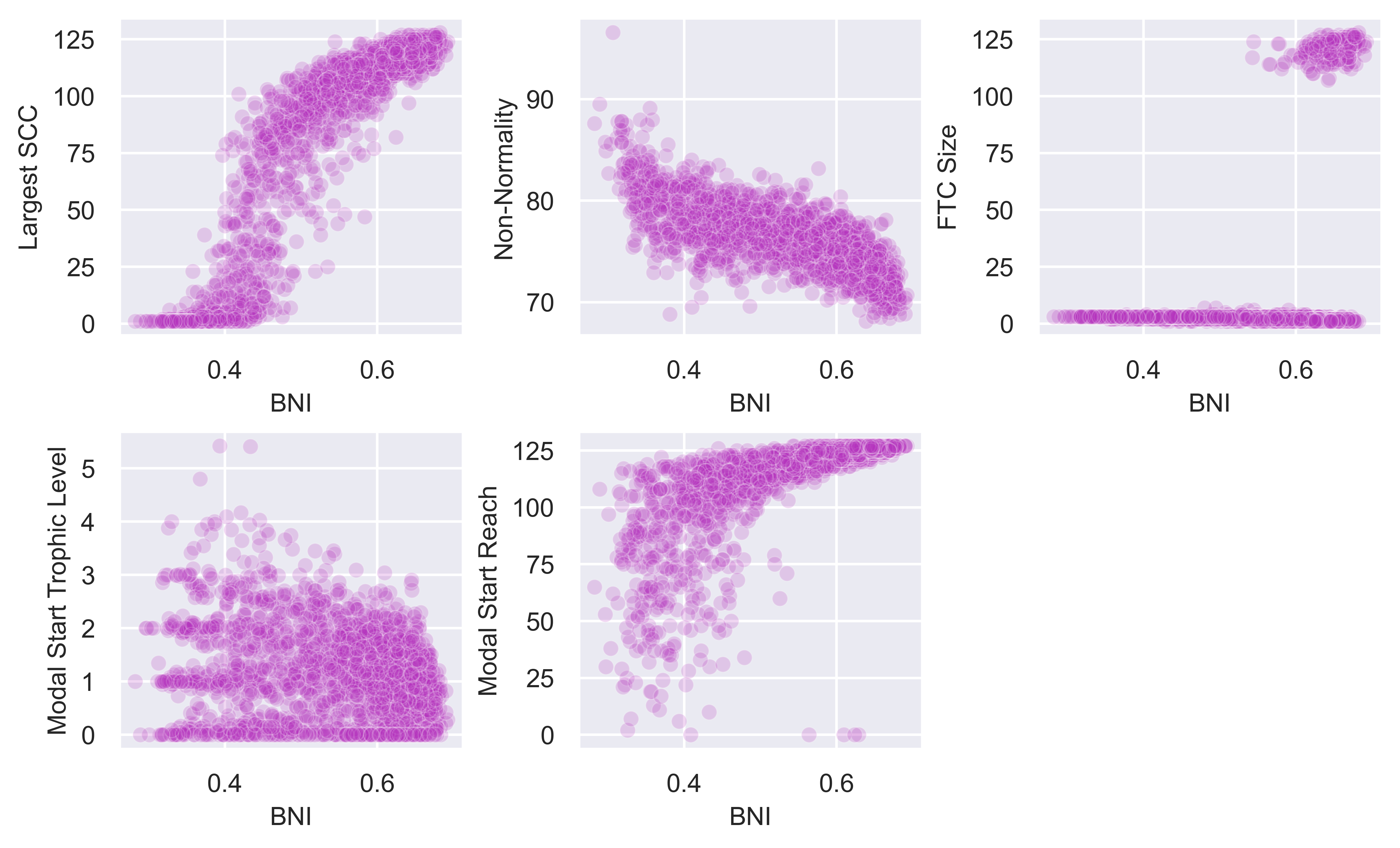}
\newcaption{Relationships of additional features with BNI}{
Scatter plots of largest strongly-connected component, non-normality, FTC size, and trophic level and reach of the modal seizure start node against BNI, for 2000 128-node networks sampled using the generalised preferential preying algorithm and simulated with additive coupling.}
\label{fig:AddFeatBNIScatter}
\end{figure*}

Fig~\ref{fig:AddFeatBNIScatter} shows scatter-plots demonstrating the respective relationships of largest strongly-connected component, non-normality, size of the FTC, and trophic level and reach of the modal seizure start node with BNI, in the case of 128-node networks sampled using the generalised preferential preying algorithm and simulated with additive coupling. Further pairwise scatterplots between all features for all model scenarios are included as Supplementary Figs~S2--S7. Size of the largest strongly-connected component has a strong, nonlinear positive relationship with BNI, which results in a higher Spearman's rank correlation coefficient than between spectral radius and BNI in several of the modelling scenarios (Tables~\ref{tab:corrs_128_additive}--\ref{tab:corrs_20_diffusive}, S1-S2). Non-normality is also negatively correlated with BNI. Size of the strongly-connected component, non-normality, spectral radius and trophic incoherence are noticeably correlated in all modelling scenarios (Supplementary Figs~S2-S7; Tables~\ref{tab:corrs_128_additive}--\ref{tab:corrs_20_diffusive}, S1-S2); these quantities encapsulate related concepts of cyclicity, strong-connectedness and directional asymmetries.

In all network sampling scenarios there is a cluster of networks with a large FTC, which is close to the size of the largest strongly-connected component (Supplementary Figs~S2-S7); these tend to have a high BNI. Since FTC is a union of strongly-connected components, a sufficiently large FTC is likely to include the largest strongly-connected component, therefore this reflects the relationship between BNI and size of the largest strongly-connected component.

Networks for which the node on which the first seizure activity most frequently emerges has high trophic level or small reach are generally less likely to attain high BNI (Supplementary Figs~S2-S7); such a node does not have influence over many other nodes and its behaviour is therefore unable to propagate to a large proportion of the network~\cite{rodgers_influence_2023}.

BNI calculated using the additively-coupled model has a stronger correlation with trophic incoherence and spectral radius than for the diffusively-coupled model (Tables~\ref{tab:corrs_128_additive}\&\ref{tab:corrs_128_diffusive},~\ref{tab:corrs_20_additive}\&\ref{tab:corrs_20_diffusive}). This supports the analytical relationship between spectral radius and model stability in the case of additive coupling, outlined in methods.

\subsection{Spectral properties of the directed adjacency matrix} \label{sec:adjacency}

We now prove that the spectral radius of the directed adjacency matrix -- which as a consequence of the Perron-Frobenius Theorem is equal to its maximally real eigenvalue -- is zero in the case of directed acyclic graphs, one in the case of graphs containing non-overlapping cycles, and greater than one in the case of graphs containing overlapping cycles, resulting in the distribution evident in the histogram displayed in Fig~\ref{fig:spectralradius}, which depicts the distribution of spectral radii for the sample of 128-node networks drawn for the additive coupling simulations.
In Methods we relate the spectral radius of the adjacency matrix to the linear stability of the network in the case of additive coupling.

\begin{figure}[!h]    \centering\hspace*{-2.5mm}\includegraphics[width=0.5\textwidth]{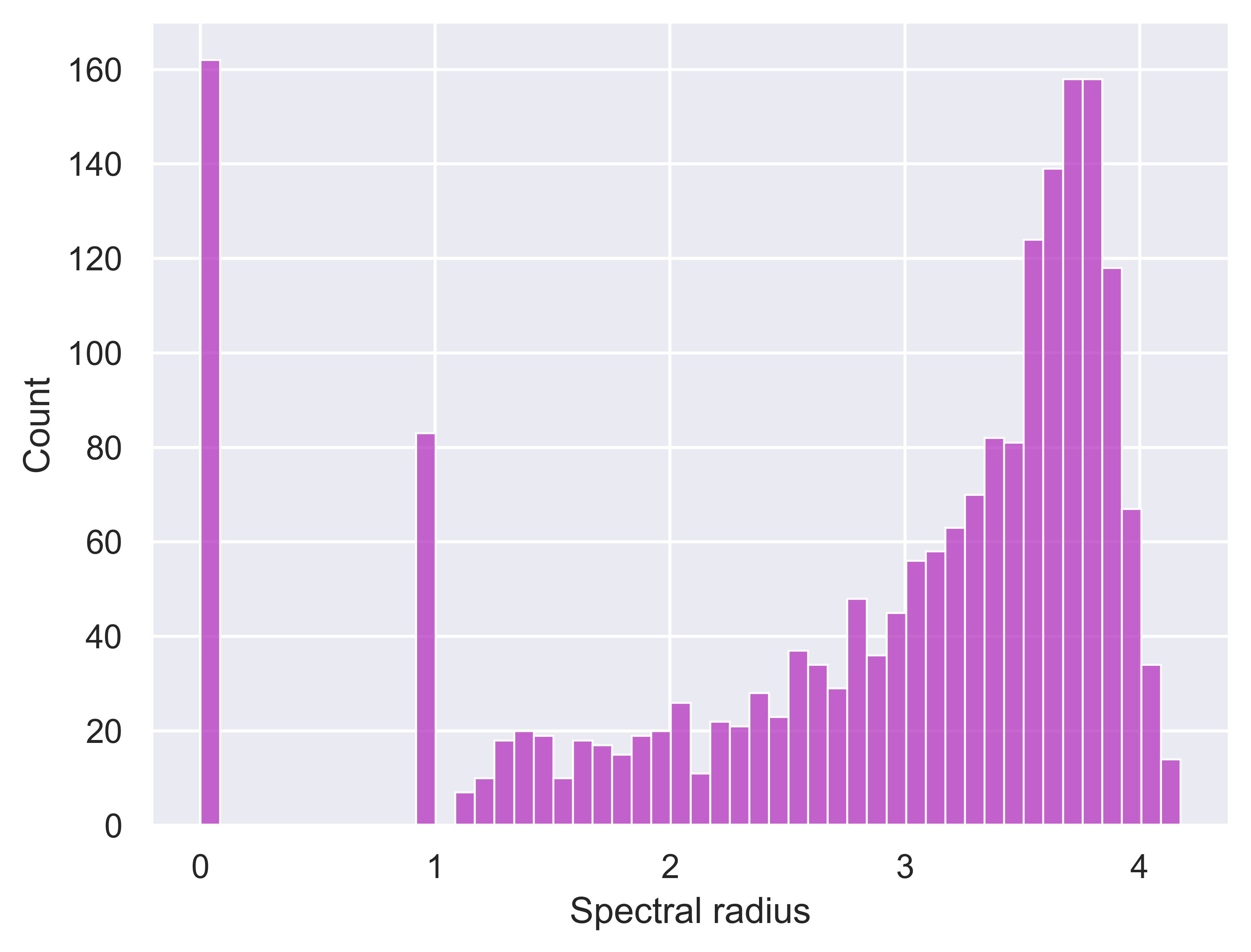}
    \newcaption{Histogram of adjacency spectral radius values}{Frequencies of values of the spectral radius of the adjacency matrix for the 2000 128-node networks of mean degree 4.0 sampled using generalised preferential preying, for the additive coupling simulations.}
    \label{fig:spectralradius}\end{figure}

We specify here the adjacency matrix of a directed graph with $N$ vertices to be an $N\times N$ matrix $A$, with entries $a_{ij}$ such that

\begin{equation}
    a_{ij}=
    \begin{cases}
        1, &\text{if there exists a directed edge } i \rightarrow j \\
        0, &\text{otherwise.}
    \end{cases}
\end{equation}

We say that cycles ``overlap'' if they have at least one vertex in common. Note that cycles are by definition strongly-connected, and a strongly-connected component with $k>1$ vertices comprises either a single $k-$cycle, or multiple overlapping cycles.

\begin{theorem*}
    Let $G$ be a directed graph with adjacency matrix $A$. Let $\rho_A$ be the spectral radius of $A$.
    \begin{itemize}
        \item[(i)] If $G$ is a directed acyclic graph (DAG), $\rho_A=0$.
        \item[(ii)] If $G$ contains at least one cycle, but no two cycles which overlap, $\rho_A=1$.
        \item[(iii)] If $G$ contains at least one pair of overlapping cycles, $\rho_A>1$.
    \end{itemize}
\end{theorem*}

\begin{proof}[Proof of (i)]
    Consider $A^k$, $k\in\mathbb{N}$, with $(i,j)^{th}$ entries equal to the number of distinct walks of length $k$ from vertex $i$ to vertex $j$.

    In a DAG, no walk intersects with the same vertex twice, since doing so would constitute a cycle. Therefore, there are no walks of length $N$ (or longer) between any pair of vertices, hence $A^N=0$.

    $A$ therefore satisfies the definition of a nilpotent matrix, which is equivalent to its only complex eigenvalue being zero.
\end{proof}

\begin{proof}[Proof of (ii)]
    The matrix $A$ is irreducible if and only if $G$ is strongly-connected. We can permute the vertices of $G$ such that $A$ is in block upper-triangular form, with irreducible diagonal blocks which represent the adjacency matrices of the strongly-connected components,

    \begin{equation}
        \tilde{A} = \begin{pmatrix}
            B_1 & * & ... & * \\
            0 & B_2 & ... & * \\
            ... & ... & ... & ... \\
            0 & 0 & ... & B_M
        \end{pmatrix},
    \end{equation}

    where $B_i$ are either $1\times 1$ matrices with entry zero, corresponding to individual vertices, or irreducible matrices, corresponding to strongly-connected components of two or more vertices.

    Then, the spectrum of $\tilde{A}$, and equivalently the spectrum of $A$, is the union of the spectra of the irreducible diagonal blocks. That is, the adjacency spectrum of $G$ is the union of the adjacency spectra of its strongly-connected components.
    
    In case \textit{(ii)}, this is the union of the adjacency spectra of isolated vertices, and isolated cycles. Isolated vertices ($B_i=\left[0\right]$) trivially have eigenvalue zero, whilst $k-$cycles are represented by circulant matrices with eigenvalues the $k^{th}$ roots of unity. Therefore, in this case, all eigenvalues have absolute value either 0 or 1, and eigenvalues of absolute value 1 exist: the spectral radius is 1.
\end{proof}

\begin{proof}[Proof of (iii)]
    It now suffices to prove that the deletion of edges from a strongly-connected digraph strictly decreases the spectral radius, since a graph with a spectral radius of 1, by part \textit{(ii)}, can be obtained through the deletion of edges to ensure that only cycles remain which do not overlap.

    The Perron-Frobenius theorem (as stated in~\cite{brouwer_spectra_2012}, Theorem 2.2.1 (v))  has, for an irreducible matrix $T \geq 0$ with spectral radius $\theta_0$,
    ``If $0 \leq S \leq T$ or if $S$ is a principal minor of $T$, and $S$ has eigenvalue $\sigma$, then $|\sigma| \leq \theta_0$; if $|\sigma|=\theta_0$, then $S=T$.''

    Note that the condition $T\geq 0$  (i.e., all entries of $T$ are nonnegative) is true for any adjacency matrix defined as above, and that $S$ is not required to be irreducible.

    The deletion of edges from a strongly-connected component with irreducible adjacency matrix $B$ produces an adjacency matrix $B'$ such that $0 \leq B' \leq B$ and $B' \neq B$. Therefore by the above result, $\rho_{B'}<\rho_{B}$.
\end{proof}

\section{Discussion}

We demonstrate in a model of epileptic seizure that directed networks with properties linked to more strongly-connected and less hierarchical architecture, such as high trophic incoherence, high adjacency spectral radius, a large strongly-connected component and low non-normality, are more susceptible to seizure.

In networks with this type of organisation, a node on which a seizure originates is more likely to have influence over large parts of the network; additionally, more complex strongly-connected structures encourage feedback around interacting cycles through which high-amplitude behaviour may accumulate.

These measures are themselves correlated, reflecting a degree of conceptual overlap and some redundancy when considering all in the same network analysis. Out of those selected, trophic incoherence, in particular, whilst underexplored in prior epilepsy literature, shows promise as a predictor of seizure propensity in directed networks. These results might also be more generically translated to other dynamical processes spreading on networks.

Our results translate across different coupling types. This suggests these aspects of network organisation have a relevance which transcends modelling choices, though some consequences of these choices still affect the extent and nature of the relationships observed.

In larger networks, the relationships between network features and simulation outcomes are stronger; smaller networks allow for less discrimination between different structures in simulation results.

We show a link between the spectral radius of the adjacency matrix and the dynamical stability of the deterministic component of the model for some modelling choices, and present a result relating the adjacency spectral radius to cycles and strong-connectivity in networks. We show in the networks simulated that this quantity is strongly correlated with trophic incoherence. Either may be chosen as a useful network metric in different contexts.

Based on the findings of this \textit{in silico} study, we recommend trophic incoherence, or related quantities such as those examined in this work, as network features to be considered in network analysis of epilepsy and other dynamical network processes. The network features analysed here may be useful as potential biomarkers. 

We also make use of spectral features of the directed adjacency matrix, as distinct from those of the Laplacian, noting the potential advantages of the former in some modelling cases. We note that the findings in this work have potential applicability to network models in a variety of contexts beyond epilepsy.

\section{Methods}
\subsection{A phenomenological model} \label{sec:model}

To model the propagation of seizure activity in the brain, we use a model describing a noise-induced transition between between a state of low-amplitude noise, representing the activity of the brain outside of the seizure state, and high-amplitude oscillations representing seizure-like activity~\cite{benjamin_phenomenological_2012,kalitzin_stimulation-based_2010}.
This model, devised to phenomenologically represent dynamics recorded in EEG at the onset of seizures, has been used for computational modelling of epileptic seizure dynamics in a variety of contexts~\cite{benjamin_phenomenological_2012,hebbink_phenomenological_2017,junges_epilepsy_2020,harrington_treatment_2024}.

The ability to transition between the two states is modelled using a modified form of a subcritical Hopf bifurcation, given as,

\begin{equation}\label{eq:bifurcation}
    \frac{\textup{d}z}{\textup{d}t} = f(z,\eta) = z\left ( \eta - 1 + \textup{i} \omega + 2\left | z \right |^{2} - \left | z \right |^{4} \right ),
\end{equation}

where $z$ is the signal representing brain activity, $\eta$ is a parameter representing excitability, which governs the bifurcation as shown in Fig~\ref{fig:bifurcation}, and $\omega$ is the natural frequency of the limit cycle.

\begin{figure}[!h]
    \centering\hspace*{-2.5mm}\includegraphics[width=0.5\textwidth]{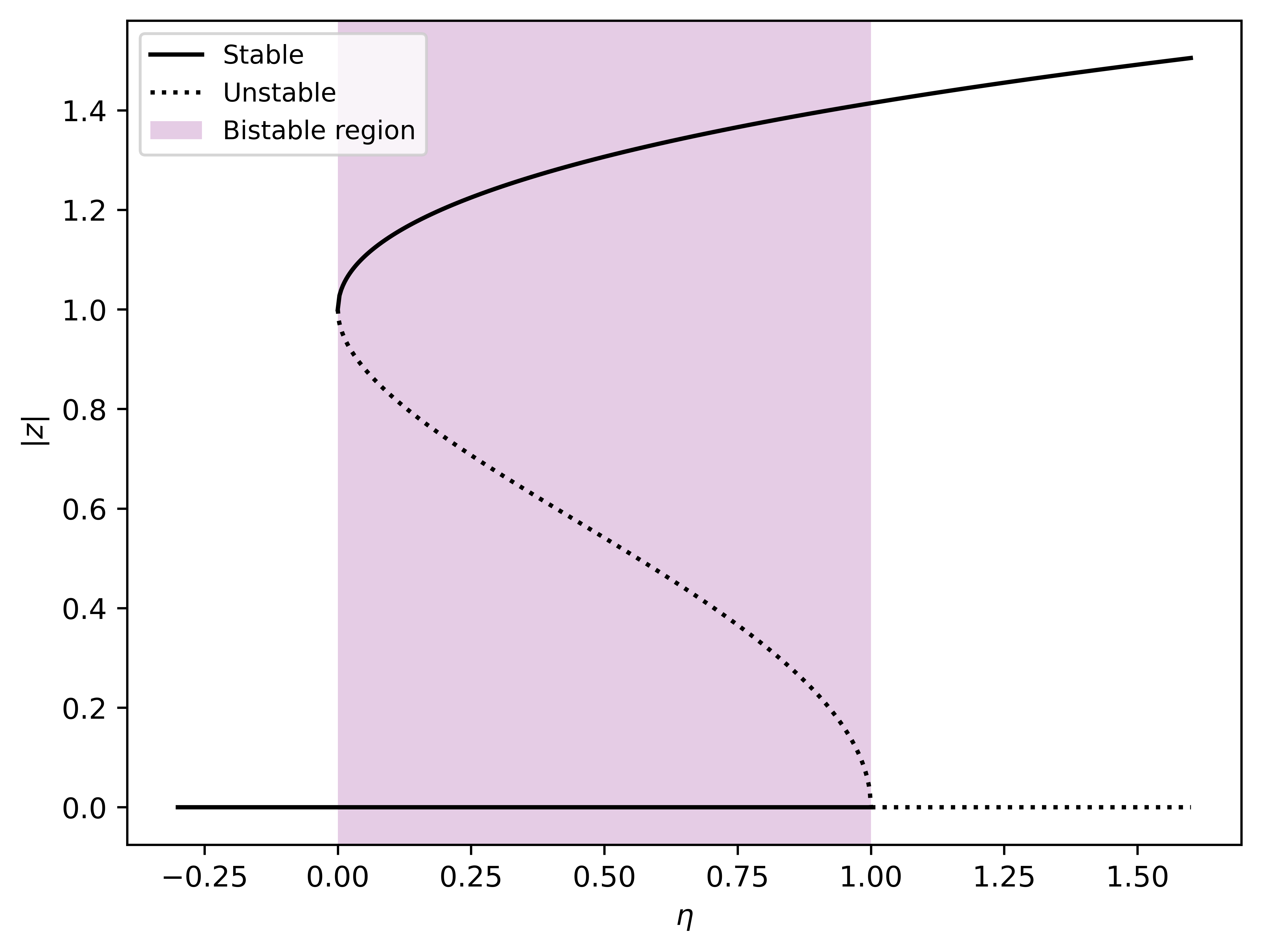}
    \newcaption{Bifurcation diagram of $|z|$ in $\eta$, in the case of Eq~(\ref{eq:bifurcation})}{Stable equilibria of $|z|$ are marked with solid lines and unstable equilibria are marked with dashed lines. An equilibrium of $|z|>0$ is a limit cycle of $z$. The model simulated in this paper modifies this bistable system into the system outlined in Eq~(\ref{eq:dz})--(\ref{eq:deta}).}
    \label{fig:bifurcation}
\end{figure}

We take as feasible values $\eta \in (0,1)$, in which cases both the steady state at $z=0$ and the limit cycle exist and are stable.

This equation is further modified to incorporate stochasticity, spontaneous seizure termination, and network dynamics, as follows.

Stochasticity is introduced in the form of a Brownian noise term $\textup{d}W_i$, scaled by a noise strength parameter $\alpha$. Through noise, representative of background brain activity, the model may be perturbed into the seizure state, allowing for spontaneous onset of seizures.

To ensure seizures terminate on a specified timescale, a transition is forced by modifying the excitability parameter $\eta$ into a time-dependent variable, simultaneously alongside $z$, with a steady state $\eta_0$ when $z=0$, which attracts the system when $z$ is in the oscillatory state. The oscillatory state of $z$ is therefore not an equilibrium of the coupled system $(|z|,\eta)$, but a transient state. An additional parameter $\tau$ determines the timescale on which the system is attracted back towards the steady state.

Finally, to simulate these dynamics on a directed network of $N$ nodes, this system is modelled for each node, with an additional coupling term such that each node $i$ is influenced by any node $j$ for which the adjacency entry $A_{ji}$ is nonzero, by some coupling function $g(z_i,z_j)$.

The complete system of equations is then,

\begin{equation}\label{eq:dz}
    \textup{d}z_{i} = \left( f(z_i,\eta_i)  + \frac{1}{N} \sum_{j=1}^{N} A_{ji} g(z_{i}, z_{j}) \right) \textup{d}t + \alpha \textup{dW}_{i}, 
\end{equation}
\begin{equation}\label{eq:deta}
    \tau \textup{d} \eta_{i} = (\eta_{0,i} - \eta_{i} - |z_{i}|^{2}) \textup{d}t ,
\end{equation}

in which coupling may, for example, be diffusive, such that,

\begin{equation}\label{eq:diffusive}
    g(z_{i}, z_{j}) = \beta (z_{j} - z_{i}),
\end{equation}

or additive, such that,

\begin{equation}\label{eq:additive}
    g(z_{j}) = \gamma z_{j},
\end{equation}

where $\beta$ and $\gamma$ are respectively the strength of each type of coupling.

Both coupling types have been used in literature studying this model~\cite{benjamin_phenomenological_2012,petkov_critical_2014,junges_epilepsy_2020,harrington_treatment_2024}. In a study comparing the two, it was found that the choice of coupling type can drastically affect modelling outcomes; authors concluded that additive coupling may be a more appropriate choice for modelling seizure dynamics~\cite{lopes_role_2023}.

\subsection{Conceptualising stability and instability}

Considering the linear stability of the deterministic component of this dynamic network model, by linearising around its fixed point $(z,\eta)=(0,\eta_0)$ we obtain a Jacobian,

\begin{equation}
\mathcal{J}=
    \begin{cases}
        \begin{pmatrix}
            \frac{\gamma}{N}A^T + H_0 - \mathbb{I}_N & 0 \\
            0 & - \mathbb{I}_N
        \end{pmatrix}
        , & g(z_{j}) = \gamma z_{j} 
        \\
        \begin{pmatrix}
            -\frac{\beta}{N}L_{in}^T + H_0 - \mathbb{I}_N & 0 \\
            0 & - \mathbb{I}_N
        \end{pmatrix}
        , & g(z_{i}, z_{j}) = \beta (z_{j} - z_{i})
        \end{cases}
\end{equation}

where $H_0=\text{diag}(\eta_0)$ is an $N\times N$ matrix with the node-wise values of $\eta_{0,i}$ as diagonal entries, $L_{in}$ is the graph Laplacian with vertex in-degrees as diagonal entries, and $\mathbb{I}_N$ is the $N\times N$ identity matrix.

In the linearised, noiseless system the excitability variable $\eta$ decouples from the activity variable $z$, contributing only the eigenvalue -1 to the spectrum of the Jacobian, and we may focus solely on the upper-left block of each of the above formulations. Indeed, the dynamics of $\eta$ govern seizure offset, while the equation for the dynamics of $z$ alone, setting $\eta=\eta_0$, is the component of the system devised to model seizure onset~\cite{benjamin_phenomenological_2012,terry_seizure_2012}. 

The linear stability of the deterministic component of the model is therefore governed in the additive coupling case by the maximum (in real part) eigenvalue of
$\frac{\gamma}{N}A^T + H_0 - \mathbb{I}_N$, where the addition of $H_0-\mathbb{I}_N$ would represents a translation of the spectrum along the real axis of the complex plane, in the scenario where $\eta_0$ has the same value on all nodes. Therefore, the graph structure contributes to the linear stability through the maximum eigenvalue of $A^T$, and similarly of $-L^T$ in the diffusive coupling case.

The maximum eigenvalue of $-L^T$ is always zero, whereas the maximum eigenvalue of $A^T$ is the spectral radius of $A$. It is therefore expected that the spectral radius of the adjacency matrix will be related to simulation outcomes, at least in the case of additive coupling.

\subsection{Trophic incoherence}

The notion of trophic incoherence in directed networks describes its ``hierarchicality''. To understand this, the concept of \textit{trophic level} for each node must be defined. This is a value which denotes a node's position in the flow of a network, such that low trophic levels are attributed to ``upstream'' nodes with influence over other nodes in the network, and high trophic levels are attributed to ``downstream'' nodes which are predominantly receivers of information in the network structure. Each edge, then, can be evaluated by the difference in trophic levels between its source and its target. The trophic incoherence of a network is then a measure of its deviation from a hierarchically-structured network in which all of these differences are equal to 1, i.e. all edges communicate upwards one trophic level.

Trophic incoherence was first proposed by Johnson et al~\cite{johnson_trophic_2014}, using a definition based on Levine's trophic levels in ecology~\cite{levine_several_1980}. 
MacKay et al revised the notions of trophic levels and trophic incoherence using the graph Laplacian, full details of which can be seen in~\cite{mackay_how_2020}. 
It is this formulation of trophic incoherence which is used in this paper. Here, 
trophic incoherence is 
measured by the mean squared deviation from 1 of the difference in trophic levels spanned by each edge.
Denoting the new definition of trophic level on node $i$ as $x_i$, the trophic incoherence of a binary network is,

\begin{equation}
    F_0 = \frac{\sum_{i,j=1}^N A_{ij}(x_j - x_i - 1)^2}{\sum_{i,j=1}^N A_{ij}}.
\end{equation}

\subsection{Additional Features}\label{sec:addfeats}

Trophic incoherence is related to the cycle structure and the spectral radius of directed networks~\cite{johnson_looplessness_2017}, as well as to the strong connectivity~\cite{rodgers_strong_2023} -- a graph theoretical feature which has previously been identified as relevant to the spread of seizure activity in network models of epilepsy~\cite{benjamin_phenomenological_2012,schmidt_dynamics_2014}. We record the size of the largest strongly-connected component for each network simulated.

The notion of hierarchy, at least in terms of the presence of ``basal'' nodes or strongly-connected regions, has previously been described for this model using the first transitive component (FTC)~\cite{benjamin_phenomenological_2012,harrington_treatment_2024}.
This is defined as the set of nodes $A$ such that for any $B$ for which there exists a path $B \to A$, there also exists a path $A \to B$. Equivalently this is the union of all basal strongly-connected components (i.e., any strongly-connected components, including those of size 1, with no in-edges originating from elsewhere). The size of the FTC has been compared with outcomes from simulations of this model in previous studies, however it has been noted that the influence of the FTC on simulation outcomes diminishes as network size increases and more complex global network dynamics dominate~\cite{harrington_treatment_2024}.

Non-normality, calculated as $||AA^T-A^TA||$, is a quantity which is related to trophic incoherence~\cite{drysdale_connection_2025}, quantifying a level of asymmetry in the adjacency matrix.

\subsection{Network simulations} \label{sec:simmethods}

For both additive and diffusive coupling, 10000 20-node networks with mean degree 2.5 and 2000 128-node node networks with mean degree 4.0, of varying trophic incoherence, were generated using the generalised preferential preying algorithm~\cite{klaise_neurons_2016}. Trophic incoherence was varied by sampling using randomly-generated values of the trophic specialisation parameter for this algorithm, generated from an exponential distribution with parameter $\theta=3.0$. For comparison, for the additively and diffusively-coupled models 10000 random 20-node networks were also generated using the Erd\H{o}s-Renyi algorithm~\cite{erdos_random_1959} with mean degree 2.5, which are expected to occupy a narrower distribution of trophic incoherence values. The smaller network size was chosen to be similar to routine, 19- or 21- channel EEG recordings from which such networks can be inferred. All networks were required to be at least weakly connected.

For each network generated, the network model was simulated in Julia over a range of values of coupling strength, with five iterations of Brownian noise for each, and parameters as described in Table~\ref{tab:modelparams}, based on methodology from previous work using this model~\cite{harrington_treatment_2024,junges_epilepsy_2020,woldman_evolving_2019}. Excitability values are chosen for each coupling type to result in simulations with nontrivial seizure activity (i.e., not dominated by seizure or non-seizure states).

\begin{table*}[!ht]
\centering
\newcaption{Parameter values used in model simulations}{}
\begin{tabular}{|l+l|l|}
\hline
\textbf{Parameter} & \textbf{Description} & \textbf{Value/s}\\ \thickhline
$T$ & Duration of simulation & $2000.0$ s \\ \hline
$\eta_0$ & Baseline excitability & 0.8 (additive), 0.9 (diffusive) \\ \hline
$\delta t$ & Simulation timestep & $0.0005$ s \\ \hline
$\alpha$ & Noise strength & $0.0535$ \\ \hline
$\beta$, $\gamma$ & Coupling strength & $1.0, 1.5, 2.0, ..., 6.0$ \\ \hline
$\tau$ & $\eta$-slowing parameter & $5.0$ s \\ \hline
$\omega$ & Natural frequency of oscillatory state & $20.0$ Hz \\ \hline
\end{tabular}
\begin{flushleft}
\end{flushleft}
\label{tab:modelparams}
\end{table*}

We measure the propensity for seizures to generate in a network using the brain network ictogenicity \cite{goodfellow_estimation_2016,junges_epilepsy_2020}, representing the proportion of the simulation which is in the seizure state, for which at least two vertices are required to have a signal of sufficiently high amplitude:

\begin{equation}
    \widehat{BNI} = \sum_{t \in \{0, \delta t, 2\delta t, ..., T \}} \frac{f(m)}{TN}
\end{equation}

where $m$ is the number of vertices on which $|z|^2>0.25$ and

\begin{equation}
    f(m) = \begin{cases}
    0, &m=1 \\
    m\delta t, &m \geq 2.
    \end{cases}
\end{equation}

\section*{Acknowledgements}

We would like to thank Dr Wessel Woldman, of Neuronostics Ltd and formerly the University of Birmingham, for his support.

The computations described in this paper were performed using the University of Birmingham’s BlueBEAR HPC service, which provides a High Performance Computing Service to the University’s research community. Further details can be found at http://www.birmingham.ac.uk/bear.

\section*{Funding}

P.K. acknowledges the support of EPSRC grant EP/V520275/1.

\section*{Author contributions}

PK: Conceptualization, Software, Formal Analysis, Visualization, Writing–original draft, Writing–review and editing. CD: Conceptualization, Writing–review and editing, Supervision. SJ: Conceptualization, Writing–review and editing, Supervision.

\section*{Code availability}

Julia code written for this project is available on Github at \hyperlink{https://github.com/pjkissack/TrophicAnalysis.jl.git}{https://github.com/pjkissack/TrophicAnalysis.jl.git}

\printbibliography

\clearpage

\setcounter{figure}{0}
\setcounter{table}{0}
\renewcommand\thefigure{S\arabic{figure}}
\renewcommand\thetable{S\arabic{table}}

\onecolumn

\section*{Supplementary Tables and Figures}

\vspace*{3em}

\begin{figure*}[h]
    \centering
    \includegraphics[width=6in]{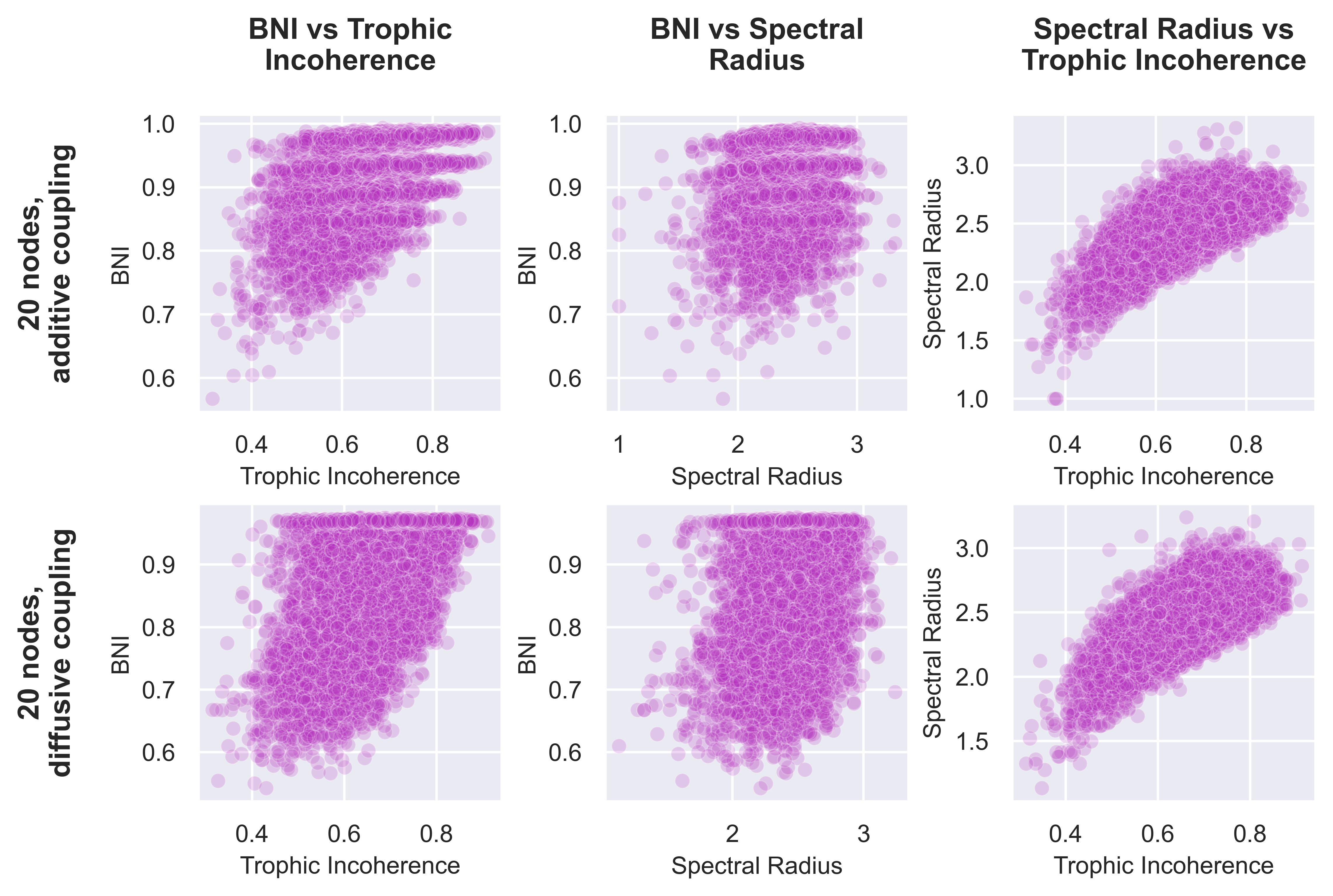}
    \newcaption{Trophic Incoherence, Spectral Radius and BNI in Erd\H{o}s-Renyi networks}{Pairwise scatter plots of
trophic incoherence, BNI and spectral radius, for 20-node networks with mean degree 2.5 sampled
using the Erd\H{o}s-Renyi algorithm and simulated with additive and
diffusive coupling.}
    \label{fig:S1}
\end{figure*}

\begin{table*}[t]
    \centering
    \begin{tabular}{|l+l|l|l|l|l|l|l|}
\hline
 & A & B & C & D & E & F & G  \\
\thickhline
A. BNI &&&&&&& \\ \hline
B. spectral radius & 0.205 &&&&&&
\\ \hline
C. trophic incoherence & 0.533 & 0.689 &&&&&
\\ \hline
D. largest strongly-connected component & 0.629 & 0.253 & 0.698 &&&&
\\ \hline
E. non normality & -0.399 & -0.465 & -0.689 & -0.474 &&&
\\ \hline
F. first transitive component size & 0.285 & 0.023 & 0.100 & 0.165 & -0.070 &&
\\ \hline
G. trophic level of the modal start node & -0.470 & -0.306 & -0.394 & -0.305 & 0.192 & -0.228 & 
\\ \hline
H. reach of the modal start node & 0.775 & 0.259 & 0.514 & 0.684 & -0.359 & 0.248 & -0.577 
\\ \hline

\end{tabular}
    \newcaption{Additively-coupled Erd\H{o}s-Renyi network simulation correlations}{Spearman rank correlations between graph features and simulation outcomes from 10000 20-node networks with mean degree 2.5 sampled using the Erd\H{o}s-Renyi algorithm, simulated with additive coupling. In the first column Brain Network Ictogenicity (BNI), a metric reflecting the propensity of simulated seizure-like behaviour to spread in a network, is compared against a selection of global network properties, and properties of the node in which seizure-like behaviour most frequently emerges across simulations of the network.}
    \label{tab:S1}
\end{table*}

\begin{table*}[t]
    \centering
    \begin{tabular}{|l+l|l|l|l|l|l|l|}
\hline
 & A & B & C & D & E & F & G  \\
\thickhline
A. BNI &&&&&&& \\ \hline
B. spectral radius & 0.235 &&&&&&
\\ \hline
C. trophic incoherence & 0.514 & 0.679 &&&&& 
\\ \hline
D. largest strongly-connected component & 0.612 & 0.224 & 0.692 &&&& 
\\ \hline
E. non normality & -0.346 & -0.473 & -0.696 & -0.468 &&& 
\\ \hline
F. first transitive component size & 0.298 & 0.017 & 0.093 & 0.143 & -0.054 &&
\\ \hline
G. trophic level of the modal start node & 0.095 & -0.002 & -0.053 & -0.027 & 0.053 & 0.487 &
\\ \hline
H. reach of the modal start node & 0.638 & 0.055 & 0.348 & 0.521 & -0.276 & -0.281 & -0.336
\\ \hline

\end{tabular}
    \newcaption{Diffusively-coupled Erd\H{o}s-Renyi network simulation correlations}{Spearman rank correlations between graph features and simulation outcomes from 10000 20-node networks with mean degree 2.5 sampled using the Erd\H{o}s-Renyi algorithm, simulated with diffusive coupling. In the first column Brain Network Ictogenicity (BNI), a metric reflecting the propensity of simulated seizure-like behaviour to spread in a network, is compared against a selection of global network properties, and properties of the node in which seizure-like behaviour most frequently emerges across simulations of the network.}
    \label{tab:S2}
\end{table*}

\begin{figure*}[!h]
    \centering
    \hspace*{-1em}
    \includegraphics[width=7.5in]{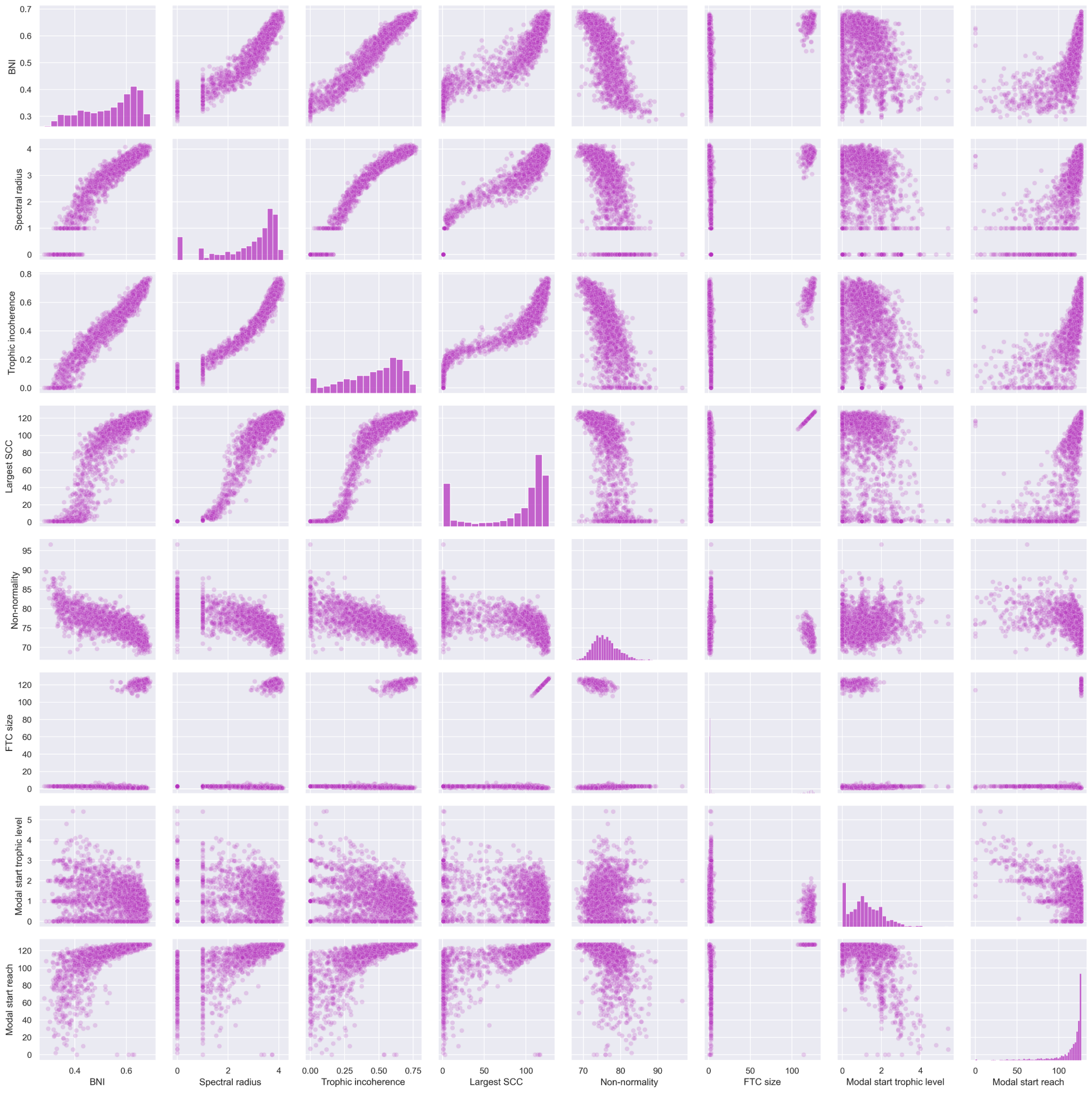}
    \newcaption{Scatter plot matrix of all features for 128-node, additively-coupled generalised preferential preying networks}{Histograms of each feature are plotted in the respective diagonal subplots.}
    \label{fig:S2}
\end{figure*}

\begin{figure*}[!h]
    \centering
    \hspace*{-1em}
    \includegraphics[width=7.5in]{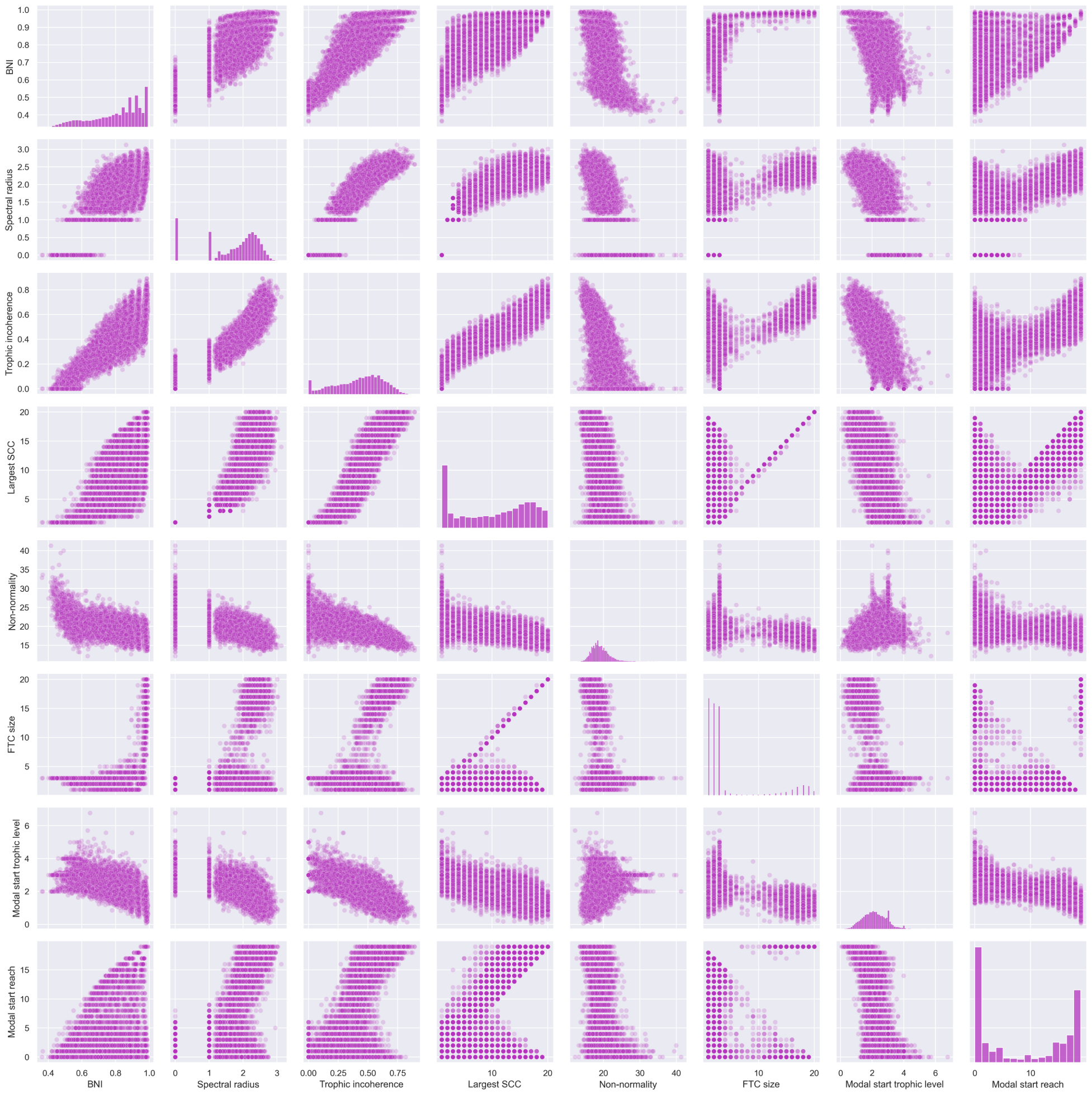}
    \newcaption{Scatter plot matrix of all features for 20-node, additively-coupled generalised preferential preying networks}{Histograms of each feature are plotted in the respective diagonal subplots.}
    \label{fig:S3}
\end{figure*}

\begin{figure*}[!h]
    \centering
    \hspace*{-1em}
    \includegraphics[width=7.5in]{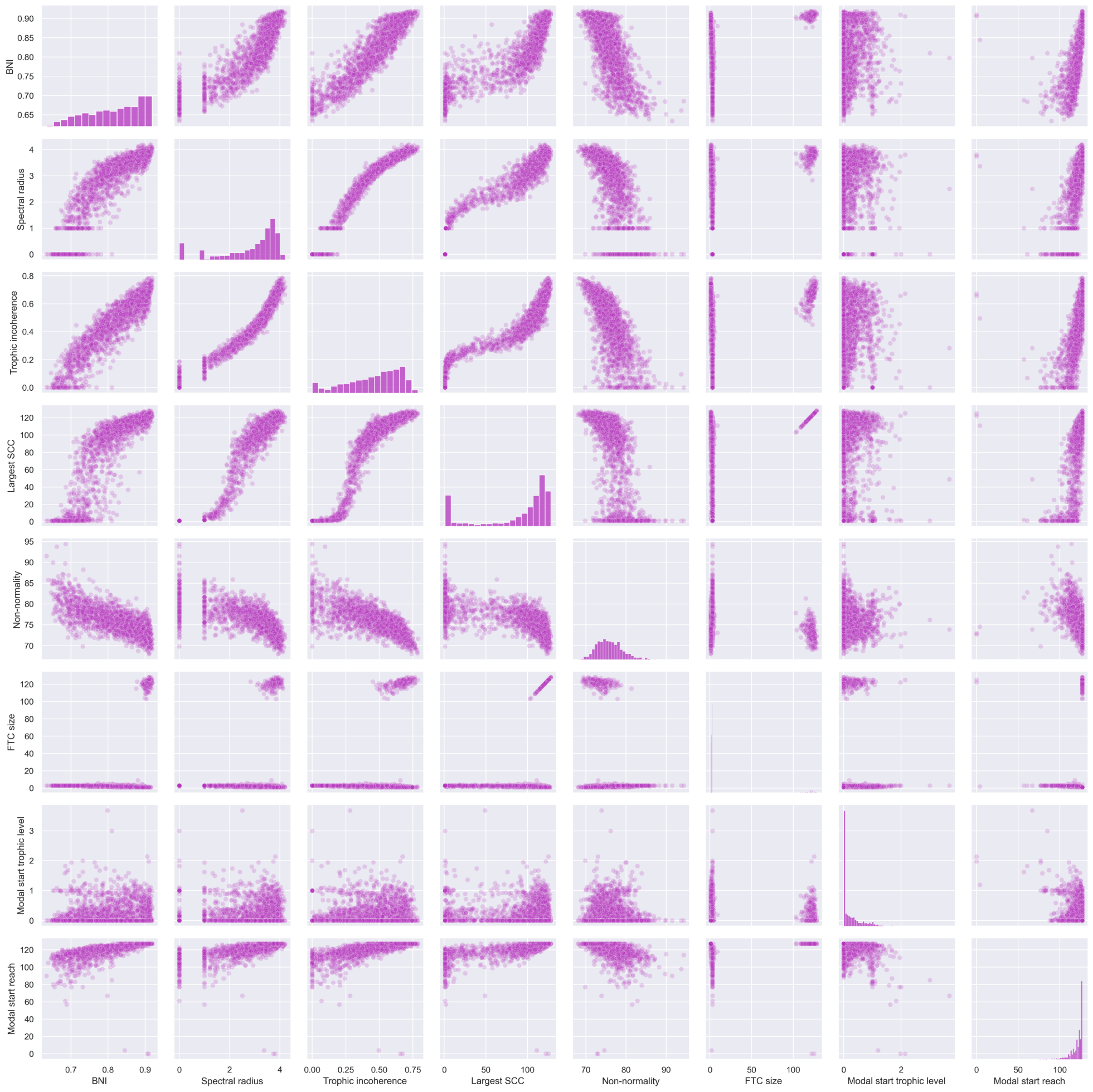}
    \newcaption{Scatter plot matrix of all features for 128-node, diffusively-coupled generalised preferential preying networks}{Histograms of each feature are plotted in the respective diagonal subplots.}
    \label{fig:S4}
\end{figure*}

\begin{figure*}[!h]
    \centering
    \hspace*{-1em}
    \includegraphics[width=7.5in]{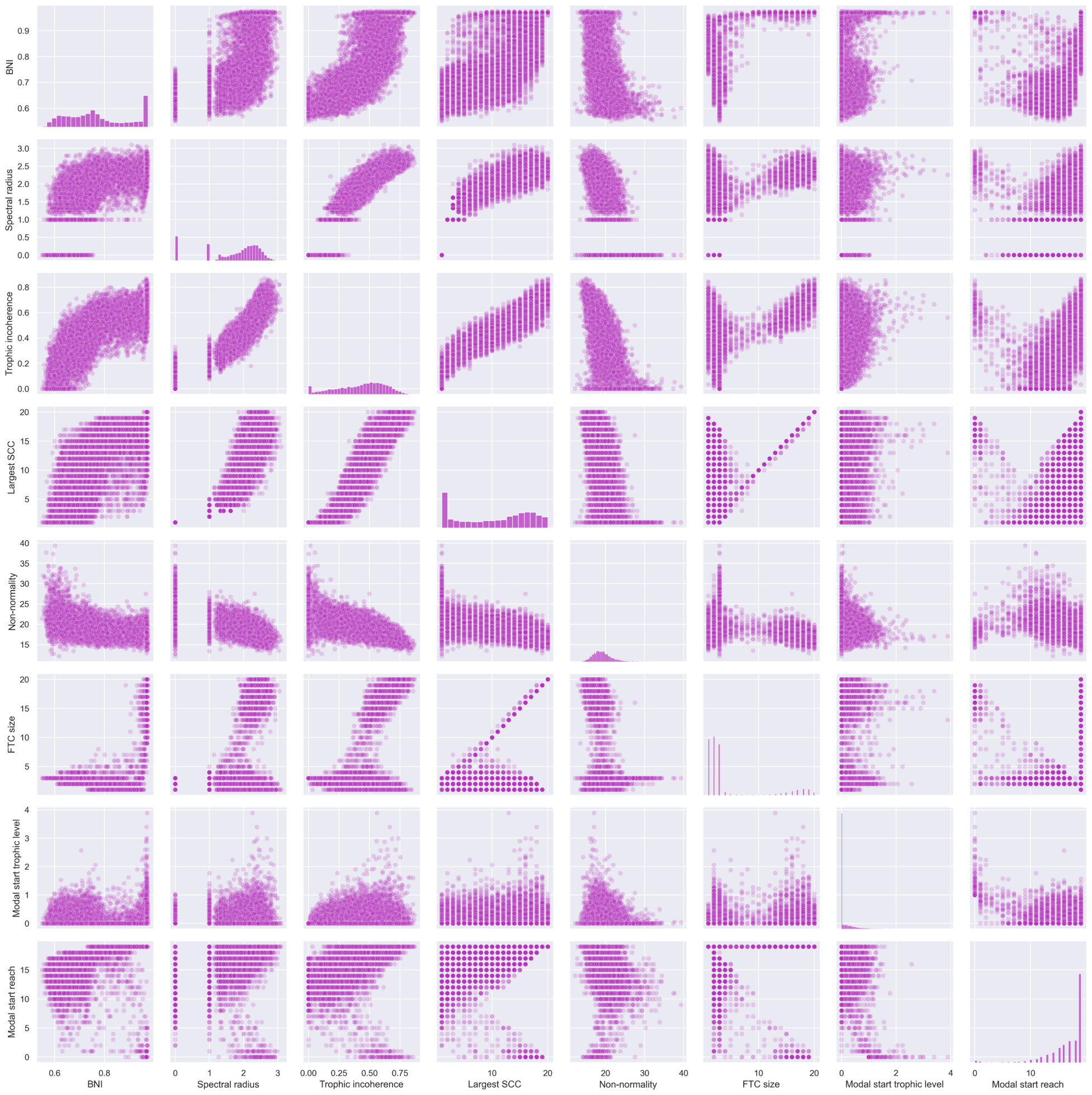}
    \newcaption{Scatter plot matrix of all features for 20-node, diffusively-coupled generalised preferential preying networks}{Histograms of each feature are plotted in the respective diagonal subplots.}
    \label{fig:S5}
\end{figure*}

\begin{figure*}[!h]
    \centering
    \hspace*{-1em}
    \includegraphics[width=7.5in]{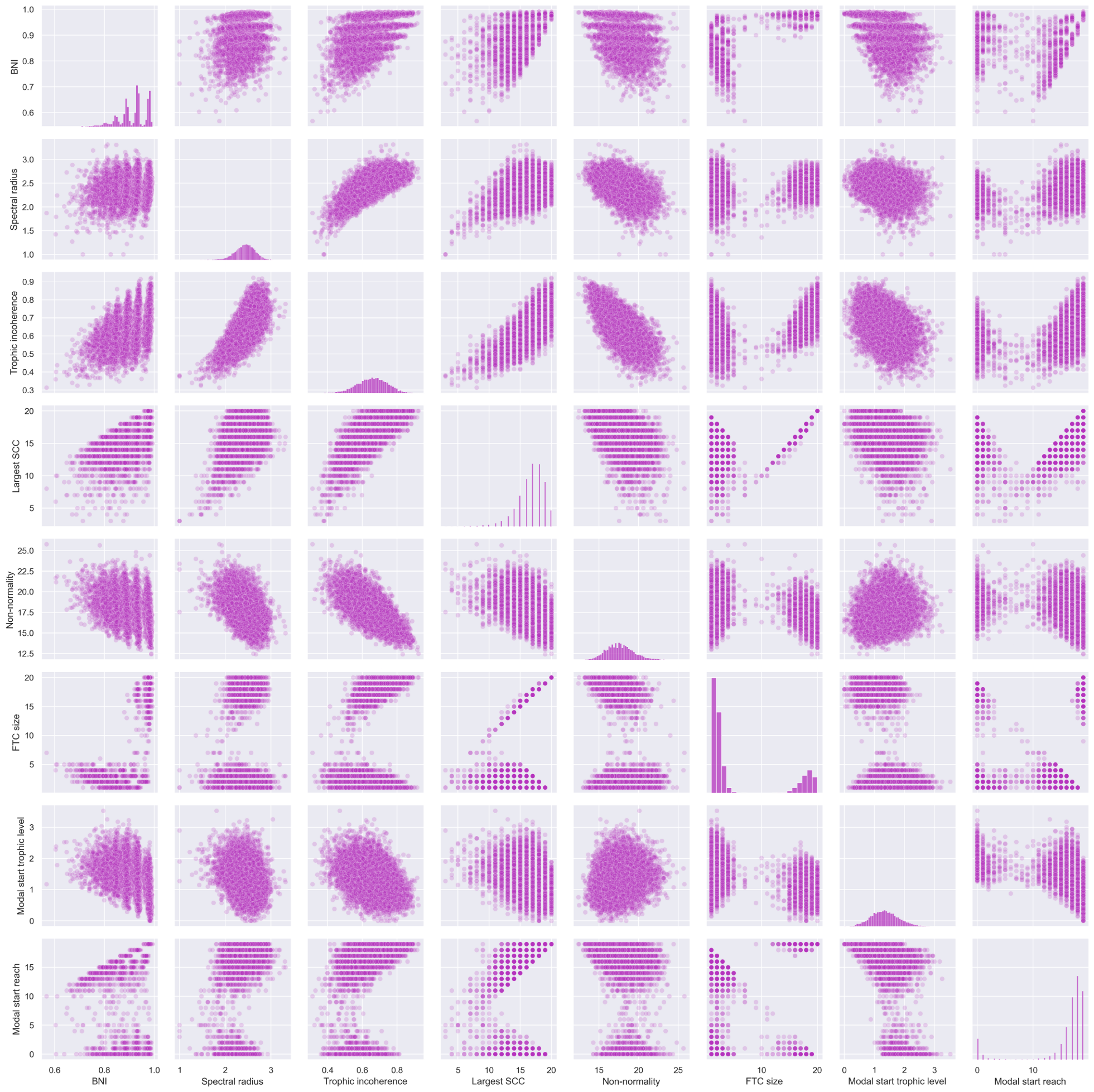}
    \newcaption{Scatter plot matrix of all features for 20-node, additively-coupled Erd\H{o}s-Renyi networks}{Histograms of each feature are plotted in the respective diagonal subplots.}
    \label{fig:S6}
\end{figure*}

\begin{figure*}[!h]
    \centering
    \hspace*{-1em}
    \includegraphics[width=7.5in]{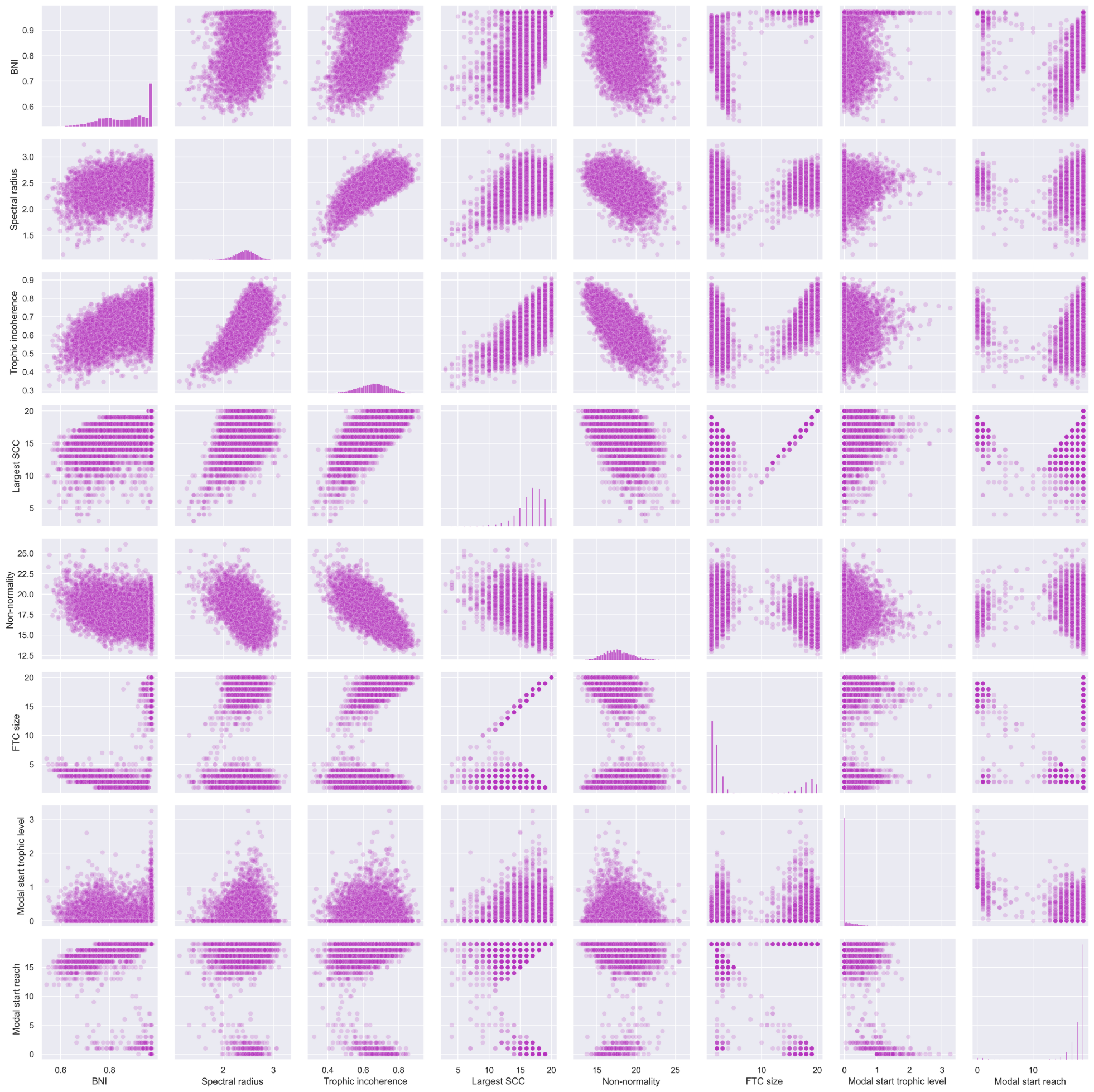}
    \newcaption{Scatter plot matrix of all features for 20-node, diffusively-coupled Erd\H{o}s-Renyi networks}{Histograms of each feature are plotted in the respective diagonal subplots.}
    \label{fig:S7}
\end{figure*}

\end{document}